%% file: mrmodel.tex
\documentclass{sig-alternate}

\pdfoutput=1
\usepackage{latexsym}
\usepackage{booktabs}
\usepackage{amsmath,amssymb}
\usepackage[boxed]{algorithm}
\usepackage{algorithmic}
\usepackage{graphicx}
\usepackage{multirow}
\usepackage[tight]{subfigure}

\usepackage[pdftex,naturalnames,pdfpagemode={UseOutlines},letterpaper,bookmarks,bookmarksopen=true,pdfstartview={FitH},bookmarksnumbered=true,pdftitle={-},pdfauthor={-}]{hyperref}

\renewcommand{\paragraph}[1]{\vspace{1mm}\noindent{\bf #1}}
\newcommand{\eat}[1]{}

\newtheorem{definition}{Definition}[section]

\newtheorem{lemma}[definition]{Lemma}

\newtheorem{theorem}[definition]{Theorem}

\newtheorem{example}[definition]{Example}

\begin{document}

\title{Upper and Lower Bounds on the Cost of a Map-Reduce Computation}

\numberofauthors{1}
\author{Foto N. Afrati$^{\dag}$, Anish Das Sarma$^{\sharp}$, Semih Salihoglu$^{\ddag}$, Jeffrey D. Ullman$^{\ddag}$\\
\affaddr{\large $^{\dag}$National Technical University of Athens,  $^{\sharp}$Google Research, $^{\ddag}$Stanford University}\\
\email{\small afrati@softlab.ece.ntua.gr, anish.dassarma@gmail.com, semih@cs.stanford.edu, ullman@gmail.com}}

\date{}

\maketitle

\input{paper}


\clearpage

\end{document}

%% file: paper.tex
\begin{abstract}
In this paper we study the tradeoff between parallelism and communication cost in
a map-reduce computation. For any problem that is not ``embarrassingly parallel,'' the finer we partition the work of the reducers so that more parallelism can be extracted, the greater will be the total communication between mappers and reducers. We introduce a model of problems that can be solved in a single round of map-reduce computation.  This model enables a generic recipe for discovering lower bounds on communication cost as a function of the maximum number of inputs that can be assigned to one reducer.   We use the model to analyze the tradeoff for three problems: finding pairs of strings at Hamming distance $d$, finding triangles and other patterns in a larger graph, and matrix multiplication.  For finding strings of Hamming distance 1, we have upper and lower bounds that match exactly.  For triangles and many other graphs, we have upper and lower bounds that are the same to within a constant factor.  For the problem of matrix multiplication, we have matching upper and lower bounds for one-round map-reduce algorithms.   We are also able to explore two-round map-reduce algorithms for matrix multiplication and show that these never have more communication, for a given reducer size, than the best one-round algorithm, and often have significantly less.
\end{abstract}

\section{Introduction}

We assume the reader is familiar with map-reduce \cite{DeanGhemawat} and its open-souce implementation Hadoop \cite{Hadoop}.  A brief summary can be found in Chapter~2 of \cite{MMDS}. There have been many custom solutions using a single round of map-reduce for specific problems, e.g., performing fuzzy joins~\cite{ADMPU12, VCL10}, clustering~\cite{TTLKF11}, graph analyses~\cite{AFU12, SV11}, multiway join~\cite{AU10}, and so on.  Here, we develop techniques for analyzing problems of this type and optimizing the performance on any distributed computing environment by explicitly studying an inherent trade-off between {\em communication cost} and {\em parallelism}.

\subsection{Communication and Parallelism for Map-Reduce}
\label{sum-appr-subsect}

This paper offers a model that helps us analyze how suited problems are to a map-reduce solution.  We focus on two parameters that represent the tradeoff involved in designing map-reduce algorithms.

First is the amount of communication between the map phase and the reduce phase.  Often, but not always, the cost of communication is the dominant cost of a map-reduce algorithm.  To represent the communication cost, we define and study {\em replication rate}.  The replication rate of any map-reduce algorithm is the average number of key-value pairs that the mappers create from each input.

The second parameter is the ``reducer size.''  A {\em reducer}, in the sense we use the term in this paper, is a reduce-key (one of the keys that can appear in the output of the mappers) together with its list of associated values, as would be delivered to a reduce-worker.  {\em Reducer size} is the upper bound on how long the list of values can be.   For example, we may want to limit a reducer to no more input than can be processed  in main memory.   A reduce-worker may be assigned many reduce-keys and works on them one at a time.   The total computation cost of the reducers is the sum over all keys (or ``reducers'') of the computation cost of processing all the values associated with that key.\footnote{Computation cost at the mappers is not treated separately, but is incorporated into the communication cost.}

Limiting reducer size also enables more parallelism.   Small reducer sizes force us to redesign the notion of a ``key'' in order to allow more, smaller reducers, and thus allow more parallelism if enough compute nodes are available.

\subsection{How the Tradeoff Can Be Used}
\label{how-use-subsect}

Suppose we have determined that the best algorithms for a problem have replication rate $r$ and reducer size $q$, where $r=f(q)$ for some function $q$.  Look ahead to Fig.~\ref{tradeoff-fig} for an example of what such a function $f$ might look like.  In particular, be aware that $f(q)$ usually grows as $q$ shrinks.   When we try to solve an instance of this problem on a particular cluster, we must determine the true costs of execution.  For example, if we are running on EC2 \cite{EC2}, we pay particular amounts for communication and for rental of virtual processors.  The communication cost is proportional to $r$; the constant of proportionality depends on the rate EC2 charges for communication and the size of our data.  The cost of renting processors is some function of $q$.

\begin{example}
\label{cost-function-ex}
If the reducer must compare all pairs of its inputs (e.g., consider the Hamming-distance-based similarity join discussed later in Example~\ref{hd1-ex}), then the work at each reducer is $O(q^2)$, and the number of reducers is inversely proportional to $q$, so the total processor cost is proportional to $q$.  That is, the cost of solving this instance of our problem is $ar+bq$ for some constants $a$ and $b$.  Since $r=f(q)$, the cost is $af(q)+bq$.  We find the value of $q$ that minimizes this expression.
That value tells us which of the algorithms lying along the curve $r=f(q)$ should be selected for this job.\footnote{Note that typically, $f(q)$ is monotonically decreasing in $q$, so there is a minimum at some finite value of $q$.}

If we were concerned more with wall-clock time than with total computation cost, then we might add a term representing the execution time for a single reducer.  In this hypothetical example, the time to compare $\binom{q}{2}$ pairs is $O(q^2)$, so we would minimize a function of the form $af(q)+bq+cq^2$.
\end{example}

Different problems will have different functions $r=f(q)$, and they will also have different functions of $q$ that measure the computation cost.  This function may not be the linear or quadratic functions suggested in Example~\ref{cost-function-ex}.  However:

\begin{itemize}

\item
Deducing the proper function of $q$ to represent the computation cost is not harder than analyzing, theoretically or experimentally, the running time of the serial algorithm that implements the reduce-function.

\end{itemize}

\subsection{Outline of the Paper}
\label{outline-subsect}

There may be many ways to solve nontrivial problems in a single round of map-reduce.  The more parallelism you want, the more communication overhead you face due to having to replicate inputs to many reducers.  In this paper:

\begin{itemize}

\item
We offer a simple model of how inputs and outputs are related.  We show how our model can capture a varied set of problems (Section~\ref{model-sect}).

\item
We define the fundamental tradeoff between

\begin{itemize}

\item[a)] Reducer size: the maximum number of inputs that one reducer can receive, and

\item[b)] Replication rate, or average number of key-value pairs to which each input is mapped by the mappers.

\end{itemize}

\item
We study three well-known problems: {\em Hamming Dist\-ance} (Section~\ref{hd1-sect}), {\em triangle finding} (Section~\ref{triangle-sect}) and some generalizations (Section~\ref{sample-graph-sect}), and {\em matrix multiplication} (Section~\ref{mm-sect}).  In each case
there is a lower bound on the replication rate that grows as reducer size shrinks (and therefore as the parallelism grows).
Moreover, we present algorithms that match
these lower bounds for various reducer sizes.

\end{itemize}

\subsection{Related Work}
\label{related-work-subsect}

In \cite{OR11}, the optimization of theta-join implementation by map-reduce was considered from a point of view similar to what we propose here.  This paper considers only one special case of our model, where each output depends on only two inputs, and they do not deal with the matter of tradeoff between reducer size and communication. An inherent trade-off between {\em communication cost} and {\em parallelism} has been studied in different contexts, e.g., pipelined parallelism~\cite{HasanM94}; we study this trade-off for single round map-reduce jobs.

The model of \cite{KSV10} proposes that a map-reduce algorithm should limit the input size of any reducer to be asymptotically smaller than the total amount of input.  This idea is appropriate for eliminating trivial algorithms that really do all the work serially in one reducer and thus limits consideration to algorithms that we might think of as truly parallel.  However, it does not let us get into the matter of size/communication tradeoffs.

Map-reduce differs from previous parallel-computation models (e.g., PRAM) in that it interleaves sequential and parallel computation.
Thus the essential constraint on map-reduce comes not so much from the demand for parallelism, but from the limit on how much input we can expect a reducer to handle and how costly communication among processors is.  For instance, if the input is small enough, then the optimal choice is to run everything at one compute node thus minimizing communication, regardless
of the asymptotics of your algorithm.

There has been a lot of interest in handling skewed data in map-reduce (e.g., \cite{KwonBHR12,KoutrisS11}).
The work closer to our setting is \cite{KoutrisS11} where the authors propose a slight modification
to the map-reduce computational framework to allow for small amount of communication
among the mappers in order to decide how to handle skewed data.
Handling skewed data is not the focus of our paper, but the need to deal with skewed data, e.g., graphs with some nodes whose degree is higher than the limit $q$ on reducer size, will require alternative algorithms.

Our model for describing problems is closely related to the notion of data provenance~\cite{provenance}. There has also been some work~\cite{IkedaPW11, OlstonR11}  on provenance in the context of distributed workflows, including map-reduce workflows.

\section{The Model}
\label{model-sect}

The model is simple yet powerful: We can develop some quite interesting and realistic insights into the range of possible map-reduce algorithms for a problem.  For our purposes, a {\em problem} consists of:

\begin{enumerate}

\item
Sets of {\em inputs} and  {\em outputs}.

\item
A {\em mapping} from outputs to sets of inputs.  The intent is that each output depends on only the set of inputs it is mapped to.

\end{enumerate}

\noindent There are two non-obvious points about this model:

\begin{itemize}

\item
Inputs and outputs are hypothetical, in the sense that they are all the possible inputs or outputs that might be present in an instance of the problem.  Any {\em instance} of the problem will have a subset of the inputs.  We assume that an output is never made unless at least one of its inputs is present, and in many problems, we only want to make the output if {\em all} of its associated inputs are present.

\item
We need to limit ourselves to finite sets of inputs and outputs.  Thus, a finite domain or domains from which inputs are constructed is essential, and a ``problem'' is really a family of problems, one for each choice of finite domain(s).  We also require that there be a finite set of outputs associated with each choice of input domain(s).  The values that these outputs can take may be a function of the inputs on which each output depends, and we do not need to specify the domain for the output in advance.  Example~\ref{sum-ex} illustrates how the outputs can compute a function of their associated inputs.

\end{itemize}

\subsection{Examples of Problems}
\label{prob-ex-subsect}

In this section we offer several examples of common map-reduce problems and how they are modeled.

\begin{example}
\label{join-ex}
{\em Natural join} of two relations $R(A,B)$ and $S(B,C)$.  The inputs are tuples in $R$ or $S$, and the outputs are tuples with schema $(A,B,C)$.  To make this problem finite, we need to assume finite domains for attributes $A$, $B$, and $C$; say there are $N_A$, $N_B$, and $N_C$ members of these domains, respectively.

Then there are $N_AN_BN_C$ outputs, each corresponding to a triple $(a,b,c)$.  Each output is mapped to a set of two inputs.  One is the tuple $R(a,b)$ from relation $R$ and the other is the tuple $S(b,c)$ from relation $S$.  The number of inputs is $N_AN_B + N_BN_C$.

Notice that in an instance of the join problem, not all the inputs will be present.  That is, the relations $R$ and $S$ will be subsets of all the possible tuples, and the output will be those triples $(a,b,c)$ such that both $R(a,b)$ and $S(b,c)$ are actually present in the input instance.
\end{example}

\begin{example}
\label{triangle-ex}
{\em Finding triangles}.  We are given a graph as input and want to find all  triples of nodes such that in the graph there are edges between each pair of these three nodes.  To model this problem, we need to assume a domain of size $N$ for the nodes of the input graph.  An output is thus a set of three nodes, and an input is a set of two nodes.  The output $\{u,v,w\}$ is mapped to the set of three inputs $\{u,v\}$, $\{u,w\}$, and $\{v,w\}$.  Notice that, unlike the previous and next examples, here, an output is a set of more than two inputs.  In an instance of the triangles problem, some of the possible edges will be present, and the outputs produced will be those such that all three edges to which the output is mapped are present.
\end{example}

\begin{example}
\label{hd1-ex}
{\em Hamming distance 1}.  The inputs are binary strings, and since domains must be finite, we shall assume that these strings have a fixed length b.  There are thus $2^b$ inputs.  The outputs are pairs of inputs that are at Hamming distance 1; that is, the inputs differ in exactly one bit.  Hence there are $(b/2)2^b$ outputs, since each of the $2^b$ inputs is Hamming distance 1 from exactly $b$ other inputs~-- those that differ in exactly one of the $b$ bits.  However, that observation counts every pair of inputs at distance 1 twice, which is why we must divide by 2.
\end{example}

\begin{example}
\label{sum-ex}
{\em Grouping and aggregation}.
This example illustrates how to deal with a problem where the outputs are more than ``yes'' or ``no'' responses to whether a given set of inputs exists.  Here, each output depends on a large set of possible inputs, and the result of an output is calculated from those of its associated inputs that actually appear in the data set.
Suppose we have a relation $R(A,B)$ and we want to implement  group-by-and-sum:

\begin{center}\begin{tabular}{l}
{\tt SELECT A, SUM(B)}\\
{\tt FROM R}\\
{\tt GROUP BY A;}\\
\end{tabular}\end{center}
We must assume finite domains for $A$ and $B$. An output is a value of $A$, say $a$, chosen from the finite domain of $A$-values, together with the sum of all the $B$-values.  This output is associated with a large set of inputs: all tuples with $A$-value $a$ and any $B$-value from the finite domain of $B$.  In any instance of this problem, we do not expect that all these tuples will be present for a given $A$-value, $a$, but (unlike the previous examples) as long as at least one of them is present there will be an output for this value $a$.
\end{example}

\subsection{Mapping Schemas}
\label{rr-subsect}

\eat{
As discussed in Section~\ref{sum-appr-subsect}, for many problems, there is a tradeoff between reducer size and replication rate.  It can be argued that the existence of such a tradeoff is tantamount to the problem being ``not embarrassingly parallel''; that is, the more parallelism we introduce, the greater will be the total cost of computation.
}

\eat{
As communication tends to be expensive, and in fact is often the dominant cost, we'd like to keep the replication rate low.  However, there is also a good reason to want to keep reducer size  low.  Doing so may enable us to execute the reduce-function in main memory.  Also, the smaller the reducer size, the more parallelism there can be and the lower will be the wall-clock time for executing the map-reduce job (assuming there is an adequate number of compute-nodes to execute all reducers in parallel).  Alas, for problems that are not embarrassingly parallel, we must expect that the smaller the reducer size, the greater will be the replication rate, that is, the number of reducers to which each input is sent.
}

In our discussion, we shall use the convention that $q$ is the maximum number of inputs that can be sent to any one reducer.  

A {\em mapping schema} for a given problem, with a given value of $q$, is an assignment of a set of inputs to each reducer, subject to the constraints that:

\begin{enumerate}

\item
No reducer is assigned more than $q$ inputs.

\item
For every output, there is (at least) one reducer that is assigned all of the inputs for that output.  We say such a reducer {\em covers} the output.  This reducer need not be unique, and it is, of course, permitted that these same inputs are assigned also to other reducers.

\end{enumerate}

The figure of merit for a mapping schema with a given reducer size $q$ is the replication rate, which we defined to be the average number of reducers to which an input is mapped by that schema.  Suppose that for a certain algorithm, the $i$th reducer is assigned $q_i \le q$ inputs, and let $I$ be the number of different inputs.  Then the replication rate $r$ for this algorithm is
$$r = \sum_{i=1}^p q_i/I$$

\begin{example}
\label{word-count-ex}
To see one subtlety of the model, consider the canonical example of a map-reduce algorithm: word-count.  In the standard formulation, inputs are documents, and the outputs are pairs consisting of a word $w$ and a count of the number of times $w$ appears among all the documents.  The standard algorithm works as follows.  The map function takes a document, breaks it into words, and for each word $w$, it generates a key-value pair $(w,1)$.  There is one reducer for each key (i.e., for each word), and the reduce-function sums the 1's in the list of values it is given for a word and thus computes the count for that word.

It looks like there is a great deal of replication, because each input results in as many key-value pairs as there are words.  However, this view is deceptive.  We could just as well have thought of the inputs as the word occurrences themselves, and then each word occurrence results in exactly one key-value pair.  That is, the replication rate is 1, independent of the limit $q$ on reducer size.\footnote{Technically, if $q$ is smaller than the number of occurrences of a particular word, then this algorithm will not work at all.  But there is little reason to chose a $q$ that small.}
Since the replication rate is identically 1, there is no tradeoff at all between $q$ and replication rate; i.e., the word-count problem is embarrassingly parallel, as we knew all along.
\end{example}

We want to derive upper and  lower bounds on the minimum possible $r$, as a function of $q$, for various problems, thus demonstrating the tradeoff between high parallelism (many small reducers, so $q$ is small) and low overhead (total communication cost~-- measured by the replication rate).  

\subsection{Independence of  Inputs in the Mappers}
When we calculate bounds on the replication rate we pretend that we have an instance
of the problem where all inputs over the given domain are present.
This actually captures the nature of map-reduce computation. Normally, in a mapper, a map function turns input objects into key-value pairs independently, without knowing what else is in the input. Thus, we can take the assumption that  the mapping schema assigns inputs to processors without reference to what inputs are actually present. Consequently,
the replication rate $r$ we calculate represents the expected communication if we multiply it
by the number of inputs actually present, so $r$ is a good measure of the communication cost incurred by any instance of the problem.

Further to this point,
recall  that $q$ counts the number of potential inputs in a reducer, regardless of which inputs are actually present for an instance of the problem.  However, on the assumption that inputs are chosen independently with fixed probability, we can expect the number of actual inputs at a reducer to be $q$ times that probability, and there is a vanishingly small chance of significant deviation for large $q$.  If we know the probability of an input being present in the data is $x$, and we can tolerate $q_1$ real inputs at a reducer, then we can use $q=q_1/x$ to account for the fact that not all inputs will actually be present.

\subsection{The Recipe for Lower Bounds}
\label{sum-strat-subsect}

While upper bounds on $r$ for all problems are derived using constructive algorithms, there is a generic technique for deriving lower bounds.  Before proceeding to concrete lower bounds, we outline in this section the recipe that we use to derive all the lower bounds used in this paper.

\begin{enumerate}

\item
{\bf Deriving $g(q)$:} First, find an upper bound, $g(q)$, on the number of outputs a reducer can cover if $q$ is the number of inputs it is given.

\item
{\bf Number of Inputs and Outputs:} Count the total numbers of inputs $|I|$ and outputs $|O|$.

\item
{\bf The Inequality:} Assume there are $p$ reducers, each receiving $q_i \leq q$ inputs and covering $g(q_i)$ outputs.  Together they cover all the outputs. That is:

\begin{equation}
\label{inequality-eq}
\sum_{i=1}^pg(q_i) \geq |O|
\end{equation}

\item
{\bf Replication Rate:}
Manipulate the inequality from Equation~\ref{inequality-eq} to get a lower bound on
the replication rate, which is $\sum_{i=1}^pq_i / |I|$.

\end{enumerate}

Note that the last step above may require clever manipulation to factor out the replication rate. We have noticed that the following ``trick'' is effective in Step~(4) for all problems considered in this paper.  First, arrange to isolate a single factor $q_i$ from $g(q_i)$; that is:
\begin{equation}
\label{recipe1-eq}
\sum_{i=1}^pg(q_i) \geq |O| \Rightarrow \sum_{i=1}^p q_i \frac{g(q_i)}{q_i} \geq |O|
\end{equation}
\noindent Assuming $\frac{g(q_i)}{q_i}$ is monotonically increasing in $q_i$, we can use the fact that $\forall q_i: q_i\leq q$ to obtain from Equation~\ref{recipe1-eq}:
\begin{equation}
\label{recipe2-eq}
\sum_{i=1}^p q_i \frac{g(q)}{q} \geq |O|
\end{equation}
Now, divide both sides of Equation~\ref{recipe2-eq} by the input size, to get a formula with the replication rate on the left:
\begin{equation}
\label{recipe3-eq}
r = \frac{\sum_{i=1}^p q_i}{|I|} \geq \frac{q|O|}{g(q) |I|}
\end{equation}
Equation~\ref{recipe3-eq}  gives us a lower bound on $r$.
Thus, in summary, given a particular problem, we derive our lower bounds in this paper
as follows:

\begin{itemize}

\item Suppose the instance of the problem has $|I|$ inputs and
$|O|$ outputs.

\item
We find an upper bound, $g(q)$, on the number of outputs any $q$ inputs can generate.

\item
If  $g(q)/q$ is monotonically increasing in $q$ then we can compute the replication rate
using our recipe.

\item Suppose  the maximum number of inputs any reducer can take
is $q$.
Then the replication rate is $r\geq \frac{q|O|}{g(q) |I|}$.

\end{itemize}

\subsection{Our Results}

We summarize our results in two tables.

Table~\ref{tab:lowerbounds} gives the lower bounds for each problem we obtain. The table enumerates for each problem the total number of inputs $|I|$, number of outputs $|O|$, the upper bound $g(q)$ on the number of outputs $q$ inputs can generate for each problem, and the lower bound we derived.

\begin{table*}[t]
\centering
\begin{tabular}{ | l | c | c | c |  c |}
\hline
{\bf Problem} &$|I|$& $|O|$ & $g(q)$ & Lower bound on $r$  \\ \hline\hline
{\bf Hamming-Distance-1, $b$-bit strings} & $2^b$ &  $\frac{b2^b}{2}$ & $\frac{q\log_2 q}{2}$ (Section~\ref{output-bound-subsect}) & $\frac{b}{\log_2q}$   (Section~\ref{hd-trade-subsect})\\ \hline
{\bf Triangle-Finding, $n$ nodes} & $\frac{n^2}{2}$ & $\frac{n^3}{6}$ & $\frac{\sqrt{2}}{3} q^{\frac{3}{2}}$ (Section~\ref{tri-subsect}) & $\frac{n}{\sqrt{2q}}$ (Section~\ref{tri-subsect}) \\ \hline
{\bf Sample Graphs (size $s$ nodes) in Alon} & $\binom{n}{2}$ or $m$ & $n^s$ & $q^{s/2}$ & $(\frac{n}{\sqrt{q}})^{s-2}$ or $(\sqrt{\frac{m}{q}})^{s-2}$ \\
{\bf Class in graph of $m$ edges, $n$ nodes} & & & (Section~\ref{alon-lower-subsect}) & (Sections~\ref{alon-lower-subsect} and~\ref{alon-m-subsect}) \\ \hline
{\bf 2-Paths in $n$-node graph} & $\binom{n}{2}$ & $\frac{n^3}{2}$ & $\binom{q}{2}$ (Section~\ref{2-path-lower-bound}) & $\frac{2n}{q}$ (Section~\ref{2-path-lower-bound}) \\ \hline
{\bf Multiway Join: $N$ bin. rels, $m$ vars., } & $N\binom{n}{2}$ & $\binom{n}{m}$ & $q^{\rho}$ (\cite{Atserias}) & $\frac{n^{m-2}}{q^{\rho -1}}$ (Section~\ref{multiway-lower-bound}) \\
{\bf Dom. $n$, parameter $\rho$ from~\cite{Atserias}} & & & & \\ \hline
{\bf $n\times n$ Matrix Multiplication} & $2n^2$ & $n^2$ & $\frac{q^2}{4n^2}$ (Section~\ref{mm-rr-subsect}) & $\frac{2n^2}{q}$ (Section~\ref{mm-rr-subsect}) \\ \hline
\end{tabular}
\caption{\label{tab:lowerbounds} Lower bound on replication rate $r$ for various problems in terms of number of inputs $|I|$, number of outputs $|O|$, and maximum number of inputs per reducer $q$.}
\end{table*}

Table~\ref{tab:upperbounds} gives the upper bound on the replication rate for each problem. In several cases our upper bounds are derived using multiple constructive algorithms, giving different upper bounds depending on the input parameters. Therefore, Table~\ref{tab:upperbounds} only gives a representative upper bound for each problem, with a forward reference to the section in which more detailed results are present.

\begin{table*}[t]
\centering
\begin{tabular}{ | l | c | c | c |  c |}
\hline
{\bf Problem} & Upper bound on $r$ \\ \hline\hline
{\bf Hamming-Distance-1} $b$-bit strings  & $\frac{b}{\log_2q}$ (Section~\ref{hd-upper-subsect})\\ \hline
{\bf Triangle-Finding, $n$ nodes} & ${\cal O}(\frac{n}{\sqrt{2q}})$   (Section~\ref{sparse-tri-subsection} and~\cite{AFU12,SV11}) \\ \hline
{\bf Sample Graphs (size $s$ nodes) in Alon} & ${\cal O}((\sqrt{\frac{m}{q}})^{s-2})$ (Result from~\cite{AFU12}) \\
{\bf Class in graph of $m$ edges, $n$ nodes} & \\ \hline
{\bf 2-Paths in $n$-node graph} & ${\cal O}(\frac{2n}{q})$ (Section~\ref{2-path-upper-bound}) \\ \hline
{\bf Multiway Join: $N$ rels, $m$ vars., Dom.} & {\bf Chain join:} $(n/\sqrt{q})^{N-1}$ \\  
{\bf $n$} {\bf (Section~\ref{multiway-upper-bound})} & {\bf Star join:} fact, dim. sizes $f$, $d_0$: $\frac{Nd_0(Nd_0/q)^{N-1}}{f+Nd_0}$ \\ \hline
{\bf $n\times n$ Matrix Multiplication} & $\frac{2n^2}{q}$ for $q\geq 2n^2$ (Section~\ref{mm-upper-subsect} and~\cite{MC69}) \\ \hline
\end{tabular}
\caption{\label{tab:upperbounds} Representative upper bound on the replication rate $r$ for each problem considered in this paper. This table only presents a representative upper bound, with a forward reference to the section that derives all upper bounds with constructive algorithms for each problem.}
\end{table*}

\section{The Hamming-Distance-1 Problem}
\label{hd1-sect}

We begin with the tightest result we can offer.  For the problem of finding pairs of bit strings of length $b$ that are at Hamming distance 1, we have a lower bound on the replication rate $r$ as a function of $q$, the maximum number of inputs assigned to a reducer.  This bound is essentially best possible, as we shall point to a number of mapping schemas that solve the problem and have exactly the replication rate stated in the lower bound.

\subsection{Bounding the Number of Outputs}
\label{output-bound-subsect}

As described in Section~\ref{sum-strat-subsect}, our first task is to develop a tight upper bound on the number of outputs that can be covered by a reducer of size $q$.

\begin{lemma}
\label{hd-outputs-lemma}
For the Hamming-distance-1 problem, a reducer of size $q$ can cover no more than
$(q/2)\log_2q$
outputs.
\end{lemma}

\begin{proof}
The proof is an induction on $b$, the length of the bit strings in the input.  The basis is $b=1$.  Here, there are only two strings, so $q$ is either 1 or 2.  If $q=1$, the reducer can cover no outputs.  But $(q/2)\log_2q$ is 0 when $q=1$, so the lemma holds in this case.  If $q=2$, the reducer can cover at most one output.  But $(q/2)\log_2q$ is 1 when $q=2$, so again the lemma holds.

Now let us assume the bound for $b$ and consider the case where the inputs consist of strings of length $b+1$.  Let $X$ be a set of $q$ bit strings of length $b+1$.  Let $Y$ be the subset of $X$ consisting of those strings that begin with 0, and let $Z$ be the remaining strings of $X$~-- those that begin with 1.  Suppose $Y$ and $Z$ have $y$ and $z$ members, respectively, so $q=y+z$.

An important observation is that for any string in $Y$, there is at most one string in $Z$ at Hamming distance 1.  That is, if $0w$ is in $Y$, it could be Hamming distance 1 from $1w$ in $Z$, if that string is indeed in $Z$, but there is no other string in $Z$ that could be at Hamming distance 1 from $0w$, since all strings in $Z$ start with 1.  Likewise, each string in $Z$ can be distance 1 from at most one string in $Y$.  Thus, the number of outputs with one string in $Y$ and the other in $Z$ is at most $\min(y,z)$.

So let us count the maximum number of outputs that can have their inputs within $X$.  By the inductive hypothesis, there are at most $(y/2)\log_2y$ outputs both of whose inputs are in $Y$, at most $(z/2)\log_2z$ outputs both of whose inputs are in $Z$, and, by the observation in the paragraph above, at most $\min(y,z)$ outputs with one input in each of $Y$ and $Z$.

Assume without loss of generality that $y\le z$.  Then the maximum number of strings of length $b+1$ that can be covered by a reducer with $q$ inputs is
$$\frac{y}{2} \log_2y + \frac{z}{2} \log_2z + y$$
We must show that this function is at most $(q/2)\log_2q$, or, since $q=y+z$, we need to show
\begin{equation}
\label{yz-eq}
\frac{y}{2} \log_2y + \frac{z}{2} \log_2z + y \le \frac{y+z}{2} \log_2(y+z)
\end{equation}
under the condition that $z\ge y$.

First, observe that when $y=z$, Equation~\ref{yz-eq} holds with equality.  That is, both sides become $y(\log_2y + 1)$.  Next, consider the derivatives, with respect to $z$, of the two sides of Equation~\ref{yz-eq}.  $d/dz$ of the left side is
$$\frac{1}{2}\log_2z + \frac{\log_2e}{2}$$
while the derivative of the right side is
$$\frac{1}{2}\log_2(y+z) + \frac{\log_2e}{2}$$
Since $z\ge y\ge0$, the derivative of the left side is always less than or equal to the derivative of the right side.  Thus, as $z$ grows larger than $y$, the left side remains no greater than the right.  That proves the induction step, and we may conclude the lemma.
\end{proof}

\subsection{Lower Bound for Hamming Distance 1}
\label{hd-trade-subsect}

We can use Lemma~\ref{hd-outputs-lemma} to get a lower bound on the replication rate as a function of $q$, the maximum number of inputs at a reducer.

\begin{theorem}
\label{hd-trade-thm}
For the Hamming-distance-1 problem with inputs of length $b$, the replication rate $r$ is at least $b/\log_2q$.
\end{theorem}

\begin{proof}
Suppose there are $p$ reducers, and the $i$th reducer has $q_i \le q$ inputs. We apply our four step recipe described in Section~\ref{sum-strat-subsect}:
\begin{enumerate}
\item {\bf Deriving $g(q)$:} Recall that $g(q)$ is the maximum number of  outputs a reducer can cover with q inputs. By Lemma~\ref{hd-outputs-lemma}, $g(q)=(q/2)\log_2q$

\item {\bf Number of Inputs and Outputs:} There are $2^b$ bitstrings of length b. The total number of outputs is $(b/2)2^b$. Therefore  $|I|=2^{b}$ and $|O|=(b/2)2^b$.

\item {\bf $\sum_{i=1}^pg(q_i) \geq |O|$ Inequality:} Substituting for $g(q_i)$ and $|O|$ from above:
\begin{equation}
\label{hd-trade-eq}
\sum_{i=1}^p\frac{q_i}{2}\log_2q_i \ge \frac{b}{2}2^b
\end{equation}

\item {\bf Replication Rate:} Finally we employ the manipulation trick from Section~\ref{sum-strat-subsect}, where we arrange the terms of this inequality so that the left side is the replication rate.  Recall we must separate a factor $q_i$ from other factors involving $q_i$ by replacing all other occurrences of $q_i$ on the left by the upper bound $q$.  That is, we replace $\log_2 q_i$ by $\log_2q$ on the left of Equation~\ref{hd-trade-eq}.  Since doing so can only increase the left side, the inequality continues to hold:
\begin{equation}
\label{hd-trade2-eq}
\sum_{i=1}^p\frac{q_i}{2}\log_2q \ge \frac{b}{2}2^b
\end{equation}
The replication rate is $r = \sum_{i=1}^pq_i/|I|=\sum_{i=1}^pq_i/2^b$. We can move factors in Equation~\ref{hd-trade2-eq} to get a lower bound on $r = \sum_{i=1}^pq_i/2^b \ge b/\log_2q$, which is exactly the statement of the theorem.
\end{enumerate}
\end{proof}

\subsection{Upper Bound for Hamming Distance 1}
\label{hd-upper-subsect}

There are a number of algorithms for finding pairs at Hamming distance 1 that match the lower bound of Theorem~\ref{hd-trade-thm}.  First, suppose $q=2$; that is, every reducer gets exactly 2 inputs, and is therefore responsible for at most one output.  Theorem~\ref{hd-trade-thm} says the replication rate $r$ must be at least $b/\log_22 = b$.  But in this case, every input string $w$ of length $b$ must be sent to exactly $b$ reducers~-- the reducers corresponding to the pairs consisting of $w$ and one of the $b$ inputs that are Hamming distance 1 from $w$.

There is another simple case at the other extreme.  If $q = 2^b$, then we need only one reducer, which gets all the inputs.  In that case, $r=1$.  But Theorem~\ref{hd-trade-thm} says that $r$ must be at least $b/\log_2(2^b) = 1$.

In \cite{ADMPU12}, there is an algorithm called Splitting that, for the case of Hamming distance 1 uses $2^{1+b/2}$ reducers, for some even $b$.  Half of these reducers, or $2^{b/2}$ reducers correspond to the $2^{b/2}$ possible bit strings that may be the first half of an input string.  Call these {\em Group I reducers}. The second half of the reducers correspond to the $2^{b/2}$ bit strings that may be the second half of an input.   Call these {\em Group II reducers}.  Thus, each bit string of length $b/2$ corresponds to two different reducers.

An input $w$ of length $b$ is sent to 2 reducers: the Group-I reducer that corresponds to its first $b/2$ bits, and the Group-II reducer that corresponds to its last $b/2$ bits.  Thus, each input is assigned to two reducers, and the replication rate is 2.  That also matches the lower bound of $b/\log_2(2^{b/2}) = b/(b/2) = 2$.  It is easy to observe that every pair of inputs at distance 1 is sent to some reducer in common.  These inputs must either agree in the first half of their bits, in which case they are sent to the same Group-I reducer, or they agree on the last half of their bits, in which case they are sent to the same Group-II reducer.

We can generalize the Splitting Algorithm so that for any $c>2$ such that $c$ divides $b$ evenly, we can have reducer size $2^{b/c}$ and replication rate $c$. Note that for reducer size $2^{b/c}$,  the lower bound on the replication rate is exactly \linebreak$b/\log_2(2^{b/c})=c$. We split each bit string $w$ into $c$ segments, $w_1w_2\cdots w_c$, each of length $b/c$.  We will have $c$ groups of reducers, numbered 1 through $c$. There will be $2^{b-b/c}$ reducers in each group, corresponding to each of the $2^{b-b/c}$ bit strings of length $b-b/c$. For $i=1,...,c$, we map $w$ to the Group-$i$ reducer that corresponds to bit string $w_1\cdots w_{i-1}w_{i+1}\cdots w_c$, that is, $w$ with the $i$th substring $w_i$ removed.  Thus, each input is sent to $c$ reducers, one in each of the $c$ groups, and the replication rate is $c$. Finally, we need to argue that the mapping schema solves the problem.  Any two strings $u$ and $v$ at Hamming distance 1 will disagree in only one of the $c$ segments of length $b/c$, and will agree in every other segment. If they disagree in their $i$th segments, then they will be sent to the same Group-$i$ reducer, because we map them to the Group-$i$ reducers ignoring the values in their $i$th segments. Thus, this Group-$i$ reducer will cover the output pair <$u, v$>.

\begin{figure}[htfb]
\centerline{\includegraphics[width=0.4\textwidth]{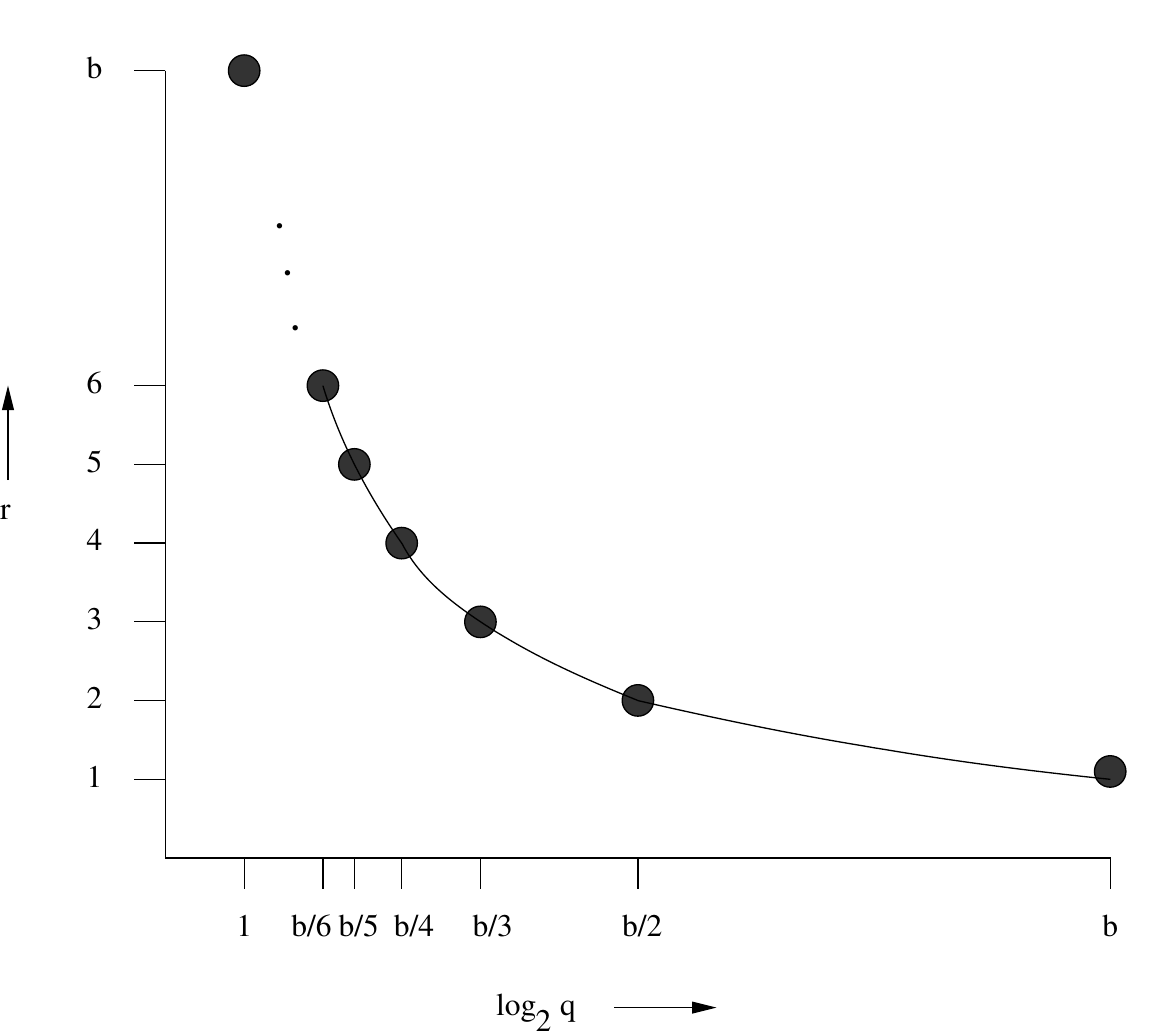}}
\caption{Known algorithms matching the lower bound on replication rate}
\label{tradeoff-fig}
\end{figure}

Figure~\ref{tradeoff-fig} illustrates what we know.  The hyperbola is the lower bound.  Known algorithms that match the lower bound on replication rate are shown with dots.

\subsection{An Algorithm for Large q}
\label{weights-subsect}

The lower bound in Fig.~\ref{tradeoff-fig} is matched for many values of $q$ as long as $\log_2q \le b/2$.  However, what happens between $b/2$ and $b$ is less clear.  Surely $r\le 2$ for that entire range.  In this subsection and the next we shall show that there are algorithms for $\log_2q$ near $b$ with replication rates strictly less than 2.

There is a family of algorithms that use reducers with large input~-- $q$ well above $2^{b/2}$, but lower that $2^b$.   The simplest version of these algorithms divides bit strings of length $b$ into left and right halves of length $b/2$ and organizes them by weights, as suggested by Fig.~\ref{weights-fig}.  The {\em weight} of a bit string is the number of 1's in that string.  In detail, for some $k$, which we assume divides $b/2$, we partition the weights into $b/(2k)$ groups, each with $k$ consecutive weights.  Thus, the first group is weights 0 through $k-1$, the second is weights $k$ through $2k-1$, and so on.  The last group has an extra weight, $b/2$, and consists of weights $\frac{b}{2}-k$ through $b/2$.

\begin{figure}[htfb]
\centerline{\includegraphics[width=0.4\textwidth]{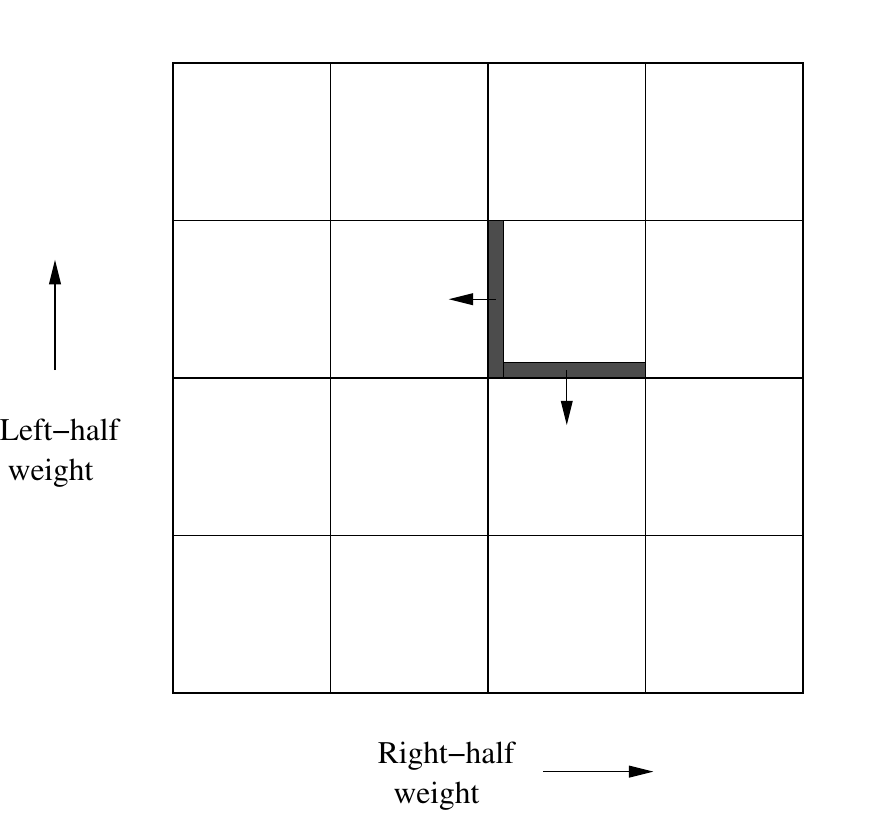}}
\caption{Partitioning by weight.  Only the border weights need to be replicated}
\label{weights-fig}
\end{figure}

There are $(\frac{b}{2k})^2$ reducers; each corresponds to a range of weights for the first half and a range of weights for the second half.  A string is assigned to reducer $(i,j)$, for $i,j=1,2,\ldots, b/2k$  if the left half of the string has weight in the range $(i-1)k$ through $ik-1$ and the  right half of the string has weight in the range $(j-1)k$ through $jk-1$.

Consider two bit strings $w_0$ and $w_1$ of length $b$ that differ in exactly one bit .  Suppose the bit in which they differ is in the left half, and suppose that $w_1$ has a 1 in that bit.  Finally, let $w_1$ be assigned to reducer $R$.  Then unless the weight of the left half of $w_1$ is the lowest weight for the left half that is assigned to reducer $R$, $w_0$ will also be at $R$, and therefore $R$ will cover the pair $\{w_0,w_1\}$.  However, if the weight of $w_1$ in its left half is the lowest possible left-half weight for $R$, then $w_0$ will be assigned to the reducer with the same range for the right half, but the next lower range for the left half.  Therefore, to make sure that $w_0$ and $w_1$ share a reducer, we need to replicate $w_1$ at the neighboring reducer that handles $w_0$.  The same problem occurs if $w_0$ and $w_1$ differ in the right half, so any string whose right half has the lowest possible weight in its range also has to be replicated at a neighboring reducer.  We suggested in Fig.~\ref{weights-fig} how the strings with weights at the two lower borders of the ranges for a reducer need to be replicated at a neighboring reducer.

Now, let us analyze the situation, including the maximum number $q$ of inputs assigned to a reducer, and the replication rate.  For the bound on $q$, note that the vast majority of the bit strings of length $n$ have weight close to $n/2$.  The number of bit strings of weight exactly $n/2$ is $\binom{n}{n/2}$.  Stirling's approximation \cite{feller68} gives us $2^n/\sqrt{2\pi n}$ for this quantity.  That is, one in $O(\sqrt{n})$ of the strings have the average weight.

If we partition strings as suggested by Fig.~\ref{weights-fig}, then the most populous $k\times k$ cell, the one that contains strings with weight $b/4$ in the first half and also weight $b/4$ in the second half, will have no more than
$$k^2\Bigr(\frac{2^{b/2}}{\sqrt{2\pi(b/2)}}\Bigl)^2 = \frac{k^22^b}{\pi b}$$
strings assigned.\footnote{Note that many of the cells have many fewer strings assigned, and in fact a large fraction of the strings have weights within $\sqrt{b}$ of $b/4$ in both their left halves and right halves.  In the best implementation, we would combine the cells with relatively small population at a single compute node, in order to equalize the work at each node.}
If $k$ is a constant, then in terms of the horizontal axis in Fig.~\ref{tradeoff-fig}, this algorithm has $\log_2q$ equal to $b - \log_2b$ plus or minus a constant.  It is thus very close to the right end, but not exactly at the right end.

For the replication rate of the algorithm, if $k$ is a constant, then within any cell there is only a small ratio of variation, among  all pairs $(i,j)$ assigned to that cells, of the numbers of strings with weights $i$ and $j$ in the left and right halves, respectively.  Moreover, when we look at the total number of strings in the borders of all the cells, the differences average out, so the total number of replicated strings is very close to $(2k)/k^2 = 2/k$.  That is, a string is replicated if either its left half has a weight divisible by $k$ or its right half does.   Note that strings in the lower-left corner of a cell are replicated twice, strings of the other $2k-2$ points on the border are replicated once, and the majority of strings are not replicated at all.  We conclude that the replication rate is $1+\frac{2}{k}$.

\subsection{Generalization to  d Dimensions}
\label{weights-d-subsect}

The algorithm of Section~\ref{weights-subsect} can be generalized from 2 dimensions to $d$ dimensions.  Break bit strings of length $b$ into $d$ pieces of length $b/d$, where we assume $d$ divides $b$.  Each string of length $b$ can thus be assigned to a cell in a $d$-dimensional hypercube, based on the weights of each of its $d$ pieces.  Assume that each cell has side $k$ in each dimension, where $k$ is a constant that divides $b/d$.

The most populous cell will be the one that contains strings where each of its $d$ pieces has weight $b/(2d)$.  Again using Stirling's approximation, the number of strings assigned to this cell is
$$k^d\Bigl(\frac{2^{b/d}}{\sqrt{2\pi b/d}}\Bigr)^d = \frac{k^d2^b}{b^{d/2}(2\pi/d)^{d/2}}$$
On the assumption that $k$ is constant, the value of $\log_2q$ is
$$b-(d/2)\log_2b$$
plus or minus a constant.

To compute the replication rate, observe that every point on each of the $d$ faces of the hypercube that are at the low ends of their dimension must be replicated.  The number of points on one face is $k^{d-1}$, so the sum of the volumes of the faces is $dk^{d-1}$.  The entire volume of a cell is $k^d$, so the fraction of points that are replicated is $d/k$, and the replication rate is $1+d/k$.  Technically, we must prove that the points on the border of a cell have, on average, the same number of strings as other points in the cell.  As in Section~\ref{weights-subsect}, the border points in any dimension are those whose corresponding substring has a weight divisible by $k$.  As long as $k$ is much smaller than $b/d$, this number is close to $1/k$th of all the strings of that length.

\subsection{Larger Hamming Distances}
\label{large-hd-subsect}

Unfortunately, the analysis for Hamming distance 1 does not generalize easily to higher distances.  To see why, consider Hamming distance 2.  While for Hamming distance 1 we learned that there is an $O(q\log q)$ upper bound on the number of outputs covered by a reducer with $q$ inputs, for distance 2 this bound is much higher: $\Omega(q^2)$, at least for small $q$.  That prevents us from getting a good lower bound on replication rate.

The $\Omega(q^2)$ bound comes from an algorithm from \cite{ADMPU12} called ``Ball-2'' that creates one reducer for each string of length $b$.  For distance 2, this algorithm assigns to the reducer for string $s$ all those strings at distance 1 from $s$.  Notice that all distinct strings at distance 1 from $s$ are distance 2 from each other.  Thus, each reducer covers $\binom{b}{2}$ outputs.  Since $q=b$, each reducer covers $\binom{q}{2}$, or about $b^2/2$ outputs.

On the other hand, we can generalize the upper bound of Section~\ref{hd-upper-subsect} to distance $d$.  We divide the $b$ bits of input strings into $k$ equal-length pieces.  A reducer corresponds to a choice of $d$ of the $k$ pieces to delete and a bit string of length $b(1-d/k)$ corresponding to the $k-d$ pieces of a string that are not deleted..  An input $s$ is sent to $\binom{k}{d}$ reducers~-- those corresponding to the strings we obtain by deleting $d$ of the $k$ segments of string $s$.  Thus, the replication rate is approximately $k^d/d!$, assuming $k$ is much larger than $d$.  Again using Stirling's approximation for the factorial, this replication rate is approximately $r = (ek/d)^d$.

\section{Triangle Finding}
\label{triangle-sect}

We shall now consider the problem of finding triangles, introduced in Example~\ref{triangle-ex}.
We shall first derive a lower bound assuming that all possible edges in the data graph can be present.  That assumption follows our model, since we assume every possible output can be made, and every possible input could be present.  However, applications of triangle-finding, such as in analysis of communities in social networks are generally applied to large but sparse graphs.  As a result, we shall continue the analysis by showing how to adjust the bound $q$ on reducer size to take into account the fact that most inputs will not be present.  When we make this adjustment, we see that the lower bound we get matches, to within a constant factor, the upper bound obtained from known algorithms.

\subsection{Lower and Upper Bound for Finding Triangles}
\label{tri-subsect}

Recall that, as described in Example~\ref{triangle-ex}, the inputs are the possible edges of a graph, and the outputs are the triples of edges that form a triangle.
Suppose $n$ is the number of nodes of the input graph.  Following the recipe from Section~\ref{sum-strat-subsect}:

\begin{enumerate}

\item {\bf Deriving $g(q)$:}
We claim a reducer with $q$ inputs can cover at most  $\frac{\sqrt{2}}{3} q^{3/2}$ outputs (triangles), which happens when the reducer is sent all the edges among a set of  $k = \sqrt{2q}$ nodes.  This point was proved, to within an order of magnitude in \cite{SV11}, who in turn credit the thesis of Schank \cite{Schank07}.\footnote{What is actually proved is that among $q$ edges, you can form at most $O(q^{3/2})$ triangles.  However, picking a set of nodes and all edges among them will match this upper bound.}
Suppose we assign to a reducer all the edges among a set of $k$ nodes.  Then there are $\binom{k}{2}$ edges assigned to this reducer, or approximately $k^2/2$ edges.  Since this quantity is $q$, we have $k = \sqrt{2q}$.
The number of triangles among $k$ nodes is $\binom{k}{3}$, or approximately $k^3/6$ outputs.  In terms of $q$, the upper bound on the number of outputs is $\frac{\sqrt{2}}{3} q^{3/2}$.

\item {\bf Number of Inputs and Outputs:}
The number of inputs is $\binom{n}{2}$ or approximately $n^2/2$.  The number of outputs is $\binom{n}{3}$, or approximately $n^3/6$.

\item {\bf $\sum_{i=1}^pg(q_i) \geq |O|$ Inequality:}
So using the formulas from (1) and (2), if there are $p$ reducers each with $\leq q$ inputs:
\begin{equation}
\label{tri1-eq}
\sum_{i=1}^p\frac{\sqrt{2}}{3} q_i^{3/2} \ge n^3/6
\end{equation}
We can replace a factor of $\sqrt{q_i}$ on the left of Equation~\ref{tri1-eq} by $\sqrt{q}$, since $q\ge q_i$, and then move that factor to the denominator of the right side.  Thus,
\begin{equation}
\label{tri2-eq}
\sum_{i=1}^p\frac{\sqrt{2}}{3} q_i \ge n^3/6\sqrt{q}
\end{equation}

\item {\bf Replication Rate:}
The replication rate is $\sum_{i=1}^pq_i$ divided by the number of inputs, which is $n^2/2$ from (1).  We can manipulate Equation~\ref{tri2-eq} as per the trick in Section~\ref{sum-strat-subsect} to get
$$r = \frac{2\sum_{i=1}^p q_i}{n^2} \ge \frac{n}{\sqrt{2q}}$$

\end{enumerate}

{\bf Upper Bound:}
There are known algorithms that, to within a constant factor, match the lower bound on replication rate.  See \cite{SV11} and \cite{AFU12}.  These algorithms are stated in terms of the number of edges, $m$, rather than the number of nodes, $n$.  However, for the case $m=\binom{n}{2}$, which is what we assume when we consider all possible edges and triangles, these algorithms do in fact imply a replication rate that is $O(n/\sqrt{q})$.  We shall next consider how to modify the analysis on the assumption that the true input will consist of $m$ randomly chosen edges.

\subsection{Analysis for Sparse Data Graphs}
\label{sparse-tri-subsection}

The lower bound $r = \Omega(n/\sqrt{q})$ holds on the assumption that all edges are actually present in the input. But as we pointed out, commonly the data graph to which a triangle-finding algorithm is applied is sparse.  We shall show that, with essentially the same limitation $q$ on the number of edges that any reducer must deal with, the lower bound on replication rate can be transformed to $r = \Omega(\sqrt{m/q})$.

Suppose the data graph has $m$ of the possible $\binom{n}{2}$ edges, and that these edges are chosen randomly.  Then if we want no more than an expected value of $q$ for the number of edges input to any one reducer, we can actually assign a ``target'' $q_t = qn(n-1)/2m$ of the possible edges to one reducer and know that the expected number of edges that will actually arrive will be $q$.

We already know from Section~\ref{tri-subsect} that if we assign at most $q_t$ of the $\binom{n}{2}$ possible edges to any reducer, then the replication rate $r$ is $\Omega(n/\sqrt{q_t})$.  But on the assumption that only $m$ edges are truly present in the input, $q_t$ is $O(qn^2/m)$, from which we can conclude
$$r = \Omega(n/\sqrt{qn^2/m}) = \Omega(\sqrt{m/q})$$

This lower bound is met (to within a constant factor) by the algorithms of \cite{SV11} and \cite{AFU12} when we measure reducer size in terms of the number of edges $m$ (as these papers do), rather than in terms of the number of possible edges $\binom{n}{2}$.  There is a natural concern that a random selection of the edges will cause more than $q$ actual edges to be assigned to some of the reducers.  However, we are only claiming bounds to within a constant factor, and by lowering the target $q_t$ by, say, a factor of 2, we can make the probability that one or more reducers will get more than $q$ actual edges as low as we like for large $n$ and $m$.

\section{Finding Instances of Other Graphs}
\label{sample-graph-sect}

The analysis of Section~\ref{triangle-sect} extends to any {\em sample graph} whose instances we want to find in a larger {\em data graph}.  For each problem of this type, the sample graph is fixed, while the data graph is the input.  Previously, we looked only at the triangle as a sample graph, but we could similarly search for cycles of some length greater than 3, or for complete graphs of a certain size, or any other small graph whose instances in the data graph we wanted to find.

\subsection{The Alon Class of Sample Graphs}
\label{alon-subsect}

In \cite{Alon81}, Noga Alon analyzed the maximum number of occurrences of a sample graph that could occur in a data graph of $n$ nodes and $m$ edges.  In particular, he defined a class of graphs, which we shall call the {\em Alon class} of sample graphs.  These graphs have the property that we can partition the nodes into disjoint sets, such that the subgraph induced by each partition is either:

\begin{enumerate}

\item
A single edge between two nodes, or

\item
Contains an odd-length Hamiltonian cycle.

\end{enumerate}
The sample graph may have any other edges as well. The Alon class is very rich. Every cycle, every graph with a perfect matching, and every complete graph is in the Alon class. Paths of odd length are also in the Alon class, since we may use alternating edges along the path as a decomposition.  However, paths of even length are not in the Alon class, since there are no cycles of any length, and the odd number of nodes cannot be partitioned into disjoint edges.

\subsection{Lower Bound for the Alon Class}
\label{alon-lower-subsect}

The key result from \cite{Alon81} that we need is that for any sample graph $S$ in the Alon class, if $S$ has $s$ nodes, then the number of instances of $S$ in a graph of $m$ edges is $O(m^{s/2})$.  So if the $i$th reducer has $q_i$ inputs, the number of instances of $S$ that it can find is $O(q_i^{s/2})$.  But if all edges are present, the number of instances of $S$ is $\Omega(n^s)$.  Note the number of instances need not be exactly $n^s$, since there may be symmetries in $S$ as we saw for the case of the triangle.  However, there are surely at least $n^s/s!$ distinct sets of nodes that form the sample graph $S$.

Now, we can repeat the analysis we did for the triangle.  If there are $p$ reducers, and the $i$th reducer has $q_i$ inputs, then
$$\sum_{i=1}^p q_i^{s/2} = \Omega(n^s)$$
If $q$ is an upper bound on $q_i$, we can write the above as
$$\sum_{i=1}^p q_i q^{(s/2)-1} = \Omega(n^s)$$

The number of inputs is $\binom{n}{2}$.
Thus, the replication rate $r$ is
$$r = \frac{\sum_{i=1}^p q_i}{\binom{n}{2}} = \Omega(n^{s-2}/q^{(s-2)/2}) = \Omega\bigl((n/\sqrt{q})^{s-2}\bigr)$$

\subsection{Bounds in Terms of Edges}
\label{alon-m-subsect}

As we did for triangles, we can scale $q$ up by a factor of $n^2/m$ if we assume that the actual data is $m$ out of the $\binom{n}{2}$ possible edges.  If we do so, the lower bound on $r$ becomes
$$r =\Omega\Bigl(\bigl(n/\sqrt{(qn^2/m)}\bigr)^{s-2}\Bigr) = \Omega\bigl((\sqrt{m/q})^{s-2}\bigr)$$
The algorithm given in \cite{AFU12} matches this lower bound, to within a constant factor.

\subsection{Paths of Length Two}
\label{2-path-subsect}

The analysis for sample graphs not in the Alon class is harder, and we shall not try to give a general rule.  However, to see the problems that arise, we will look at the simplest non-Alon graph: the path of length 2 ({\em 2-paths}).  The problem of finding 2-paths is similar, although not identical to, the problem of computing a natural self join
$$E(A,B)\bowtie E(B,C)$$
The difference is that the edge relation $E$ contains sets of two nodes, rather than ordered pairs.  That is, if a tuple $(u,v)$ is in $E$, when finding 2-paths we can treat it as $(v,u)$, even if the latter tuple is not found in $E$.

\subsubsection{Lower Bound}
\label{2-path-lower-bound}

We again follow the recipe from Section~\ref{sum-strat-subsect}:

\begin{enumerate}
\item {\bf Deriving $g(q)$:}
Any two distinct edges can be combined to form at most one 2-path.  Thus, the number of outputs (2-paths) covered by this reducer is at most $\binom{q}{2}$ or approximately $q^2/2$.

\item {\bf Number of Inputs and Outputs:}
$|I|$ is $\binom{n}{2}$ or approximately $n^2/2$. For counting $|O|$, observe that there are $\binom{n}{3}$ sets of three nodes, and any three nodes can form a 2-path in three ways.  That is, any of the three nodes can be chosen to be the middle node.  Thus, $|O|$ is $3\binom{n}{3}$, or approximately $3n^3/6 = n^3/2$.

\item {\bf $\sum_{i=1}^pg(q_i) \geq |O|$ Inequality:}
Using the formulas from (1) and (2), if there are $p$ reducers each with $\leq q$ inputs:
\begin{equation}
\label{2path1-eq}
\sum_{i=1}^pq_i^2/2 \ge n^3/2
\end{equation}
Replacing a factor of $q_i$ by $q$ on the left:
\begin{equation}
\label{2path2-eq}
\sum_{i=1}^p (q_i)(q/2) \ge n^3/2
\end{equation}

\item {\bf Replication Rate:}
We rearrange terms in Equation~\ref{2path2-eq} to make the left side equal to $\sum_{i=1}^pq_i$ divided by $|I|=n^2/2$.
$$r = \frac{\sum_{i=1}^p q_i}{n^2/2} \ge 2n/q$$
\end{enumerate}

This lower bound on replication rate is unlike those we have seen before.  For small $q$ it makes sense, but for $q>2n$ it is less than 1, which is useless.  Rather, it should be replaced by the trivial lower bound $r\ge1$ for large $n$.  Once we make this replacement, the bound is tight for an infinite number of pairs of $q$ and $n$.
If $q=n^2/2$, then we can send all edges to one reducer and do the work there, so $r=1$ is correct.

\subsubsection{Upper Bound}
\label{2-path-upper-bound}

If $q=n$, then we can have one reducer for each node.  We send the edge $(a,b)$ to the reducers for its two nodes $a$ and $b$.  The replication rate is thus 2, which agrees with the lower bound.  The reducer for node $u$ receives all edges consisting of $u$ and another node, so it can put them together in all possible ways and produce all 2-paths that have $u$ as the middle node.

If $q<n$, we have to divide the task of producing the 2-paths with middle node $u$ among several different reducers.  That means every pair of edges with $u$ as one end has to be assigned to some reducer in common.  Suppose for convenience that $k^2$ divides $n$.  Suppose $h$ is a hash function that divides the $n$ nodes into $k$ equal-sized buckets.  The reducers will correspond to pairs $[u,\{i,j\}]$, where $u$ is a node (intended to be the node in the middle of the 2-path), and $i$ and $j$ are bucket numbers in the range $1,2,\ldots,k$.  There are thus $n\binom{k}{2}$ or approximately $nk^2/2$ reducers.

Let $(a,b)$ be an edge.  We send this edge to the $2(k-1)$ reducers $[b,\{h(a),*\}]$ and $[a,\{*,h(b)\}]$, where $*$ denotes any bucket number from 1 to $k$ other than the other bucket number in the set.  We claim that any 2-path is covered by at least one reducer.  In particular, look at the reducer $[u,\{i,j\}]$.  This reducer covers all 2-paths $v-u-w$ such that $h(v)$ and $h(w)$ are each either $i$ or $j$.  Note that if $h(v)=h(w)$, then many reducers will cover this 2-path, and we want only one to produce it.  So we let the reducer $[u,\{i,j\}]$ produce the 2-path
$v-u-w$ if either

\begin{enumerate}

\item
One of $h(v)$ and $h(w)$ is $i$ and the other is $j$, or

\item
$h(v)=h(w)=i $ and $j=i+1$ modulo $k$ (i.e., $j=i+1\le k$ or $i=k$ and $j=1$).

\end{enumerate}

Each reducer receives $q = 2n/k$ edges, and as mentioned, the replication rate $r$ is $2(k-1)$, or approximately $2k$.  Since $2n/q = k$, the lower bound is approximately half what this algorithm ach\-ieves.  Thus, to within a constant factor, the upper and lower bounds match for small $q$ as well as for large $q$ (where both bounds are between 1 and 2).

\subsection{Multiway Join}
\label{multiway-join-subsect}

We begin by looking at the join of several binary relations.
We can think of this extension as looking for sample graphs in a data graph with labeled edges; the relation names are the edge labels.
Suppose $n$ is the number of nodes of the data graph.  
The inputs are the possible edges of a graph, and the outputs are the sets of $s$ edges that make the body of the multiway join true (i.e., the $s$ labeled edges of the sample graph). We assume also that the multiway join seen as a Datalog rule, (or
as a hypergraph) uses $m$ variables ($m$ attributes/nodes in the hypergraph equivalently).

\subsubsection{A Lower Bound for Multiway Join}
\label{multiway-lower-bound}

Following the recipe from Section~\ref{sum-strat-subsect}:
\begin{enumerate}

\item {\bf Deriving $g(q)$:}
According to \cite{Atserias} when we have $q$ inputs in a multiway join, then
we can have at most $g(q)=q^{\rho}$ outputs where $\rho$ is a parameter that
depends on properties of the hypergraph associated with the specific multiway
join. E.g., if the hypergraph has $\rho _1$ edges that cover all the nodes,
and this is the minimum number of edges with this property, then $\rho =\rho _1$.
Otherwise, $\rho$ comes from the solution of a linear program that is associated
to the hypergraph (see \cite{Atserias} for details of how to compute $\rho$).
From here on, we drop constant factors, but do not use the implied big-oh notation, for simplicity.

\item {\bf Number of Inputs and Outputs:}
$|I|=s\binom{n}{2}$
or on the order of $n^2$.
$|O|=\binom{n}{m}$
or on the order of $n^m$.
Note that $m$ here is a constant, so in big-oh calculations we can drop the factor $1/m!$ when approximating binomial coefficients.

\item {\bf $\sum_{i=1}^pg(q_i) \geq |O|$ Inequality:}
Replacing for $g(q)$ and $|O|$ from above:
\begin{equation}
\label{tri1-eq-mj}
\Sigma_{i=1}^p q_i^{\rho} \geq n^m
\end{equation}
We can replace a factor of $q_i^{\rho -1}$ on the left of Equation~\ref{tri1-eq-mj} by $q^{\rho -1}$, since $q\geq q_i$, and then move that factor to the denominator of the right side.  Thus,
\begin{equation}
\label{tri2-eq-mj}
\Sigma_{i=1}^p q_i \geq n^m/q^{\rho -1}
\end{equation}

\item {\bf Replication Rate:}
The replication rate is $\sum_{i=1}^pq_i$ divided by the number of inputs, which is $n^m$ from (1).  We can manipulate Equation~\ref{tri2-eq-mj} as per the trick in Section~\ref{sum-strat-subsect} to get
$$r = \frac{\sum_{i=1}^p q_i}{n^2} \ge \frac{n^{m-2}}{q^{\rho -1}}$$

\end{enumerate}

This lower bound can be easily generalized from binary relations to the case where all relations
have the same arity $\alpha\ge2$. In order to have a more quantitative
picture let us assume also that $\rho=s/\alpha$ where $s$ is the number of
relational atoms in the join. Then the replication rate lower bound is:
$$r \ge \frac{n^{m-\alpha}}{q^{s/\alpha -1}}$$
To get a more quantitative picture, we can take the special case where $s=m$,
i.e., when the number of relational atoms in the join and the number of shared variables
coincide (e.g., the join corresponds to a hypertree with an additional edge).
In this case
the lower bound becomes
$r\geq n^{m-\alpha}q^{1-m/\alpha}$.

Below we discuss in detail some algorithms from the literature for chain joins
that match this lower bound.

\subsubsection{Upper Bound for Cases of Multiway Join}
\label{multiway-upper-bound}

{\bf Chains of odd number of relations}.  Suppose we have $N$ relations
in the chain and $N$ is an odd positive integer. Then, let us compute
more carefully the above lower bound. We have now $m=N+1$ and $\rho= (N+1)/2$. Hence
the lower bound is:
$$r\geq \frac{n^{N-1}}{q^{(N+1)/2 -1}}=(n/\sqrt{q})^{N-1}$$

For the upper bound we use the results in \cite{AU10}.
This paper computes the communication cost
for when we have $p$ reducers (denoted $k$ in \cite{AU10}), each relation has size $R$ and there are $N$ relations in the join,
hence the total input
size is $|I|=RN$. In \cite{AU10} the expression that gives the communication cost is given as the sum of $N$ terms, each a product of shares (denoted $a_i$'s in \cite{AU10}). The $a_i$'s are computed in there and if we do the arithmetic,
we get communication cost per input (hence replication rate) to be equal to (up to a
factor of $p^{\frac{4}{N^2-1}}$):
$r=p^{\frac{N-3}{N-1}}$. After similar arithmetic manipulations as in previous
sections, we get this upper bound on the replication rate to be:
$$r=(n/\sqrt{q})^{N-1}$$
This upper bound matches the lower bound we computed. The case for even number
of relations
in a chain query is similar with the same conclusion.

{\bf Star joins}.  A star join joins a central {\em fact table} with several {\em dimension tables}. It is expected that the fact table is very
large while the dimension tables are smaller but still large. Suppose the size of the fact table is $f$, and all dimension tables have the same size
$d_0$. Then according to \cite{AU10}, in order to minimize the communication cost, the share for the attributes not in the fact table is 1, while, the share for each attribute in the fact table is $d_0p^{1/N}/d_0=p^{1/N}$,
where $N$ is the number of dimension tables and $p$ is the number of reducers. We assume, moreover, that dimension tables pairwise
do not share attributes. Thus, using the communication cost as computed
in \cite{AU10} and dividing it
by $f+Nd_0$ we get the replication rate:
$$r=\frac{f+Nd_0p^{\frac{N-1}{N}}}{f+Nd_0}$$
To compute  $p$ in terms of the average reducer size, we take the equation
$r(f+Nd_0)=pq$, which after replacing $r$ from above and dividing by $p$ becomes:
$f/p+Nd_0p^{\frac{-1}{N}}=q$.
In order to simplify the calculations, we assume  that $f/p =(1-e)q$ with $e$ being between 0 and 1 but not very small or very large. This is a reasonable assumption, since the size of the
fact table is much larger than the sizes of the dimension tables; it tells us that a good fraction
of the input to each reducer comes from the fact table.
Then we can solve the above to get:
$p=(Nd_0/eq)^N$,
and substituting it in the replication rate:
$$r=\frac{f+Nd_0(Nd_0/eq)^{N-1}}{f+Nd_0}$$
Substituting in the enumerator $f=pq(1-e)$ we get:
$$r=\frac{e(1-e)Nd_0(Nd_0/eq)^{N-1}}{f+Nd_0}$$

{\bf A Lower Bound for Star Join}.
Since the practical applications of star join assume that the
fact table is order of magnitude larger than the dimension tables,
we make a similar assumption here. Let $N$ be the number of dimension tables.
Suppose we have in our database $n_1$ constants (values) that
are values to the attributes outside the fact table, and the arity of each
dimension table is $m=m_1+m_2$ where $m_2$ is the number of attributes that
is shared with the fact table. Then we have at most $n_1^{Nm_1}$
tuples in the output of the join. Notice that this number is in general
much less than the size $f$ of the fact table. We apply our technique for finding lower bounds as follows:

\begin{enumerate}
 
\item {\bf Deriving $g(q)$:}
The parameter $\rho$ is equal to $N$, and thus $g(q)=q^N$.

\item {\bf Number of Inputs and Outputs:}
$|I|$ is $f+ n_1^{m_1}$.  $|O|$ is
$n_1^{Nm_1}$. Notice that the number of outputs only depends on
the dimension tables' parameters.

\item  {\bf $\sum_{i=1}^pg(q_i) \geq |O|$ Inequality:}  Substituting for $g(q_i)$ and $|O|$:
\begin{equation}
\label{tri1-eq-mj-star}
\Sigma_{i=1}^p q_i^{N} \geq n_1^{Nm_1}
\end{equation}
We replace a factor of $q_i^{N-1}$ on the left of Equation~\ref{tri1-eq-mj-star} by $q^{N -1}$, and then move that factor to the denominator of the right side:
\begin{equation}
\label{tri2-eq-mj-star}
\Sigma_{i=1}^p q_i \geq n_1^{Nm_1}/q^{N -1}
\end{equation}

\item {\bf Replication Rate:}
We get from Equation~\ref{tri2-eq-mj-star}
$$r = \frac{\sum_{i=1}^p q_i}{f+ n_1^{m_1}} \ge \frac{n_1^{Nm_1}/q^{N -1}}{f+ n_1^{m_1}}$$
We can write the above inequality as:
$$r=\frac{Nd_0(Nd_0/q)^{N-1}}{f+Nd_0}$$
This differs from the lower bound we computed by a factor of $e(1-e)/e^N$, which
under the assumptions of the star join can be thought of a constant.
\end{enumerate}

{\bf Size of output of multiway join in the general case}.
For the general case, we can apply the same technique to get lower bounds on
the replication rate,
only we need to know how to compute a bound on  the size of the output of
any multiway join. We explain here how to compute a tight bound as offered
in \cite{Atserias,GroheM06}.

Let $q$ be a multiway join and let $G(q)$ be the corresponding hypergraph.
Thus the nodes of the hypergraph are the attributes in the query and the edges of the hypergraph correspond to the relational atoms in the query. For each edge $e$ of $G(q)$
we have a variable $x_e$. Let $S$ be the number of subgoals in the query and $a_e$ be the number of attributes for the relational atom corresponding to the edge $e$.  We form the following linear program:
$$\Sigma_{e\in G(q)} a_e x_e \geq S$$
$$minimize \Sigma_{e\in G(q)} x_e$$
The solution to this program is called an optimal {\em fractional edge cover}
of the query hypergraph. It can be shown \cite{GroheM06} that there is always a solution
whose values are rational and of bit-length polynomial in the size of the query.
Fractional edge covers can be used to give an upper bound on the size $|O|$ of the
output of the query. Let $|R_e|$ be the size of the relation that corresponds
to the edge $e$ of the hypergraph $G(q)$.
$$|O| \leq \Pi_{e\in G(q)} |R_e|^{x_e}$$

\section{Matrix Multiplication}
\label{mm-sect}

We shall now take up the common application of matrix multiplication.
That is, we suppose we have $n\times n$ matrices $R=[r_{ij}]$ and $S=[s_{jk}]$ and we wish to form their product $T=[t_{ik}]$, where $t_{ik} = \sum_{j=1}^n r_{ij}s_{jk}$.
This problem introduces a number of ideas not present in the previous examples.
First, each output depends on many inputs, rather than just two or three.  In particular, the output $t_{ik}$ depends on an entire row of $R$ and and entire column of $S$, that is, $2n$ inputs, as suggested by Fig.~\ref{io-matrix-fig}.

\begin{figure}[htfb]
\centerline{\includegraphics[width=0.4\textwidth]{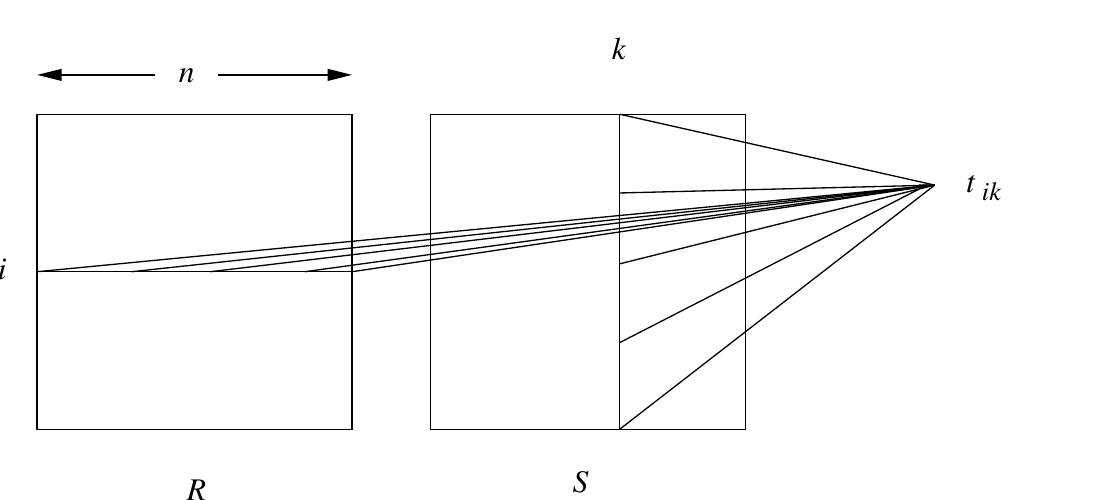}}
\caption{Input/output relationship for the matrix-multiplication problem}
\label{io-matrix-fig}
\end{figure}

There is also an interesting structure to the way outputs are related to inputs, and we can exploit that structure.  Finally, the fact that sum is associative and commutative lets us explore methods that use two interrelated rounds of map-reduce.  Surprisingly, we discover that two-round methods are never worse than one-round methods, and can be considerably better.

\subsection{The Lower Bound on Replication Rate}
\label{mm-rr-subsect}

\begin{enumerate}
\item {\bf Deriving $g(q)$:}
Suppose a reducer covers the outputs $t_{14}$ and $t_{23}$.  Then all of rows $1$ and $2$ of $R$ are input to that reducer, and all of columns $4$ and $3$ of $S$ are also inputs.  Thus, this reducer also covers outputs $t_{13}$ and $t_{24}$.  As a result, the set of outputs covered by any reducer form a ``rectangle,'' in the sense that there is some set of rows $i_1,i_2,\ldots i_w$ of $R$ and some set of columns $k_1,k_2,\ldots,k_h$ of $S$ that are input to the reducer, and the reducer covers all outputs $t_{i_uk_v}$, where $1\le u\le w$ and $1\le v\le h$.

We can assume this reducer has no other inputs, since if an input to a reducer is not part of a whole row of $R$ or column of $S$, it cannot be used in any output made by the reducer.  Thus, the number of inputs to this reducer is \hbox{$n(w+h)$}, which must be less than or equal to $q$, the upper bound on the number of inputs to a reducer.  As the total number of outputs covered is $gh$, it is easy to show that for a given $q$, the number of outputs is maximized when the rectangle is a square; that is, $w=h=q/(2n)$.  In this case, the number of outputs covered by the reducer is $g(q) = q^2/(4n^2)$.

\item {\bf Number of Inputs and Outputs:} There are two matrices each of size $n^2$. Therefore $|I|=2n^2$ and $|O|=n^2$.

\item  {\bf $\sum_{i=1}^pg(q_i) \geq |O|$ Inequality:}  Substituting for $g(q_i)$ and $|O|$:
$$\sum_{i=1}^p \frac{q_i^2}{4n^2} \ge n^2$$
\item {\bf Replication Rate:}
We first leave one factor of $q_i$ on the left as is, and replace the other factor $q_i$ by $q$.  Then, we manipulate the inequality so the expression on the left is the replication rate and obtain:
$$r = \sum_{i=1}^p \frac{q_i}{2n^2} \ge \frac{2n^2}{q}$$
\end{enumerate}

\subsection{Matching Upper Bound on Replication Rate}
\label{mm-upper-subsect}

The lower bound $r\ge 2n^2/q$ can be matched by an upper bound for a wide range of $q$'s.
If $q\ge 2n^2$, then the entire job can be done by one reducer, and if $q<2n$, then no reducer can get enough input to compute even one output.  Between these ranges, we can match the lower bound by giving each reducer a set of rows of $R$ and an equal number of columns of $S$.

The technique of computing the result of a matrix multiplication by tiling the output by squares is very old indeed \cite{MC69}.  In the map-reduce model, that is correct if a single round of map-reduce is used, but, as we shall see in Section~\ref{mm-2phase-subsect}, not quite correct for two-phase matrix multiplication, where the minimum cost occurs when the matrices are tiled with rectangles of aspect ratio 2:1.

Let $s$ be an integer that divides $n$, and let $q=2sn$.  Partition the rows of $R$ into $n/s$ groups of $s$ rows, and do the same for the columns of $S$.  There is one reducer for each pair $(G,H)$ consisting of a group $G$ of $R$'s rows and a group $H$ of $S$'s columns.  This reducer has $q=2sn$ inputs, and can produce all the outputs $t_{ik}$ such that $i$ is one of the rows in group $G$ and $k$ is one of the columns in the group $H$.
Since every pair $(i,k)$ has $i$ in some group for $R$ and has $k$ in some group for $S$, every element of the product matrix $T$ will be produced by exactly one reducer.

The replication rate for each input element is the number of groups with which its group is paired.  That number is $r = n/s$, since both $R$ and $S$ are partitioned into this number of groups.  Since $q = 2sn$, and thus $s=q/(2n)$, we have that $r = 2n^2/q$, exactly matching the lower bound on $r$.

\subsection{Matrix Multiplication Using Two Phases}
\label{mm-2phase-subsect}

There is another strategy for perfoming matrix multiplication using two map-reduce jobs.
As we shall see, this method always beats the one-phase method.  An interesting aspect of our analysis is that, while tiling by squares works best for the one-phase algorithm,

\begin{itemize}

\item
For the two-phase algorithm, the least cost occurs when the matrices are tiled with rectangles that have aspect ratio 2:1.

\end{itemize}

We assume that we are multiplying the same $n\times n$ matrices $R$ and $S$ as previously in this section.

\begin{enumerate}

\item
In the first phase, we compute $x_{ijk} = r_{ij}s_{jk}$ for each $i$, $j$, and $k$ between 1 and $n$.  We sum the $x_{ijk}$'s at a given reducer if they share common values of $i$ and $k$, thus producing a {\em partial sum} for the pair $(i,k)$.

\item
In the second phase, the partial sum for each pair $(i,k)$ is sent from each reducer that has computed at least one $x_{ijk}$ for some $j$ to a reducer of the second phase whose responsibility to to sum all these partial sums and thus compute $t_{ik}$.

\end{enumerate}
Figure~\ref{2phasemr-fig} suggests what the mappers and reducers of the two phases do.

\begin{figure}[htfb]
\centerline{\includegraphics[width=0.5\textwidth]{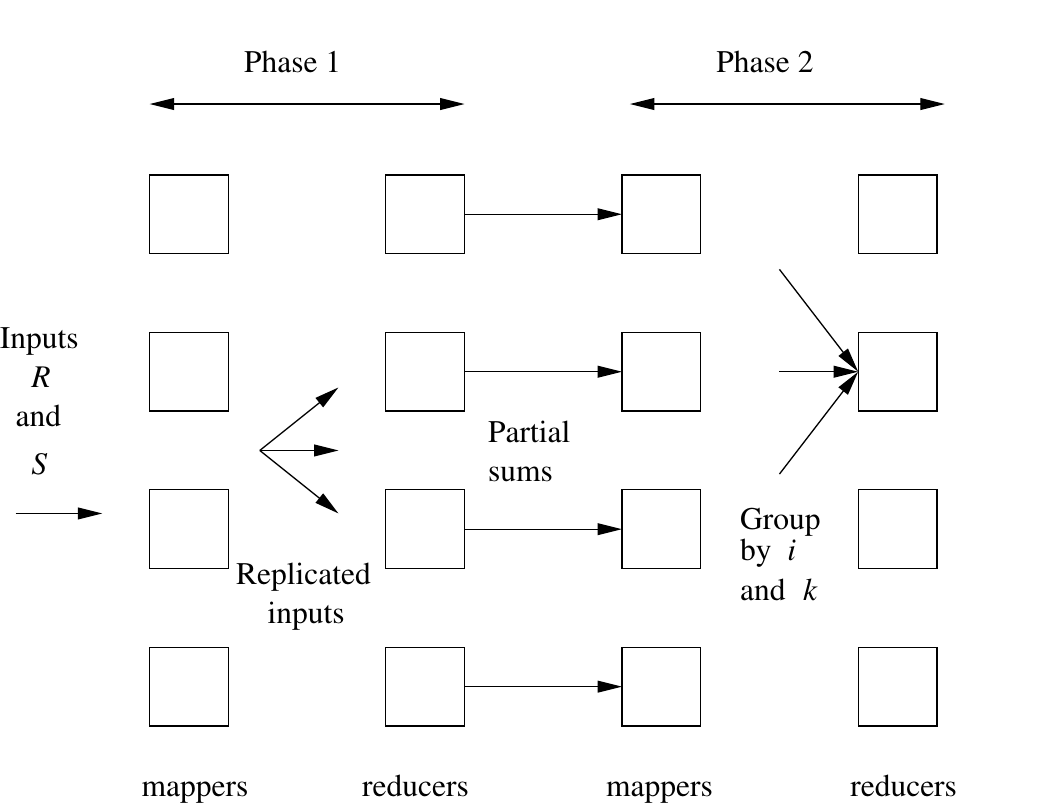}}
\caption{The two-phase method of matrix multiplication}
\label{2phasemr-fig}
\end{figure}

The second phase is embarrassingly parallel, since each partial sum contributes to only one output.  However, the first phase requires careful organization.  To begin, it is not sufficient to compute only the replication rate of the first phase, since there is significant communication in the second phase.  The number of partial sums could be as large as $n^3$ and thus dominate the communication cost.  We shall thus calculate the total communication involved in both phases.

To begin this calculation, note that the mappers of the second phase can reside at the same compute node as the $x_{ijk}$'s to which they apply.  Thus, no communication is needed between the first-phase reducers and the second-phase mappers.  The communication between the second-phase mappers and reducers is equal to the sum over all first-phase reducers of the number of different $(i,k)$ pairs for which they compute at least one $x_{ijk}$.

The communication between first-phase mappers and reducers depends on the limit $q$ we choose for the number of inputs to a reducer in the first phase\footnote{The reducers in the second phase require only $n$ inputs in the worst case, so we can ignore the input size for the second-phase reducers.}
and on the strategy we use for assigning inputs to these reducers.  As for the one-phase algorithm, we can simplify the options regarding what inputs go to what reducers by observing that the set of outputs covered by a reducer again forms a ``rectangle.'' That is, if a reducer covers both $x_{ijk}$ and $x_{yjz}$, then it also covers $x_{ijz}$ and $x_{yjk}$.  The proof is that to cover $x_{ijk}$ the reducer must have inputs $r_{ij}$ and $s_{jk}$, while to cover $x_{yjz}$ the same reducer gets inputs $r_{yj}$ and $s_{jz}$.  From these four inputs, the reducer can also cover $x_{ijz}$ and $x_{yjk}$.

We now know that the set of outputs covered by a reducer can be described for each $j$ by a set $G_j$ of row numbers of $R$ and a set of column numbers $H_j$ of $S$, such that the outputs covered are all $x_{ijk}$ for which $i$ is in $G_j$ and $k$ is in $H_j$.  As before, the greatest number of covered outputs occurs when the rectangle is a square.  That is, each reducer has an equal number of rows and columns for each $j$.  We do not know that these rows and columns must be the same for each $j$, but it is easy to argue that if not, we could reduce the communication in the first and second phases, or both, by using the same sets of rows and columns for each $j$.

Thus, we shall assume that each reducer in the first phase is given a set of $s$ rows of $R$, $s$ columns of $S$, and $t$ values of $j$ for some $s$ and $t$.  Figure~\ref{st-cube-fig} suggests how one reducer covers a cube in the three-dimensional space defined by the indexes $i$, $j$, and $k$.  There is a reducer covering each $x_{ijk}$, which means that the number of reducers is $(n/s)^2(n/t)$.   Then each element of matrices $R$ and $S$ must be sent to $n/s$ reducers, so the total communication in the first phase is $2n^3/s$.  To see why, consider an element $r_{ij}$ of matrix $R$.  We know $i$ and $j$, so only $k$ is unknown.  The number of reducers that need inputs with the particular $i$ and $j$ and some $k$ is $n/s$.  The analogous argument applies to elements of matrix $S$.

\begin{figure}[htfb]
\centerline{\includegraphics[width=0.25\textwidth]{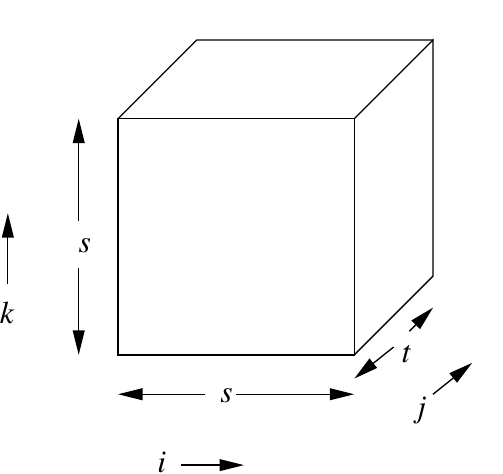}}
\caption{The responsibility of one reducer in the first phase}
\label{st-cube-fig}
\end{figure}

Each reducer produces a partial sum for $s^2$ pairs $(i,k)$.  Thus, the communication in the second phase is $s^2$ times the number of reducers, or $s^2(n/s)^2(n/t) = n^3/t$.  The sum of the communication in the first and second phases is
$$\frac{2n^3}{s} + \frac{n^3}{t}$$
We must minimize this function subject to the constraint that $2st=q$, where $q$ is the maximum number of inputs a reducer in the first phase can receive.  The reason for this constraint is that such a reducer receives $r_{ij}$ for $s$ different values of $i$ and $t$ different values of $j$, and it receives $s_{jk}$ for $s$ different values of $k$ and $t$ different values of $j$.  The method of Lagrangean multipliers lets us show that the minimum is obtained when $s=2t$.  That is, $t=\sqrt{q}/2$ and $s=\sqrt{q}$.

With these values of $s$ and $t$, the total communication is
$$\frac{2n^3}{\sqrt{q}} + \frac{n^3}{\sqrt{q}/2} = \frac{4n^3}{\sqrt{q}}$$
On the other hand, the total communication for the optimum one-phase method described in Section~\ref{mm-upper-subsect} is the replication rate times the number of inputs, or $$(2n^2/q)\times 2n^2 = 4n^4/q$$

For what values of $q$ does the one-phase method use less communication than the two-phase method?  Whenever
$$\frac{4n^4}{q} < \frac{4n^3}{\sqrt{q}}$$
or $q>n^2$.  That is, for any number of reducers except 1, the two-phase method uses less communication than the one-phase method, and for small $q$ the two-phase approach uses a lot less communication. There are other costs besides communication, of course, but since both methods perform the same arithmetic operations the same number of times, we expect that in most situations, the communication difference is decisive.

\section{Summary}
\label{summary-sect}

This paper has attempted to set a new direction for the study of optimal map-reduce algorithms.
We introduced a simple model for map-reduce algorithms, enabling us to study their performance across a spectrum of possible computing clusters and computing-cluster properties such as communication speed and main-memory size. We identified {\em replication rate} and {\em reducer input size} as two parameters representing the communication cost and compute-node capabilities, respectively, and we demonstrated that for a wide variety of problems these two parameters are related by a  precise tradeoff formula. These problems include finding bit strings at a fixed Hamming distance, finding triangles and other fixed sample graphs in a larger data graph, computing multway joins, and matrix multiplication.

\subsection{Open Problems}
\label{open-prob-subsect}

The analyses done in this paper for several problems of interest should be carried out for many other problems.  Discovering the tradeoff for Hamming distances greater than 1 seems hard.  Analogous investigations are warranted for other kinds of similarity joins besides those based on Hamming distance.   One question that arises naturally is how closely the general lower bound on multiway joins derived in this paper matches the general upper bounds in \cite{AU10}? Since there is no closed formula for either upper or lower
bounds in the general case, this question seems to need nontrivial arguments in order
to be answered.

Another interesting direction is to explore whether it is possible to analyze algorithms taking two or more rounds of map-reduce along the lines of Section~\ref{mm-2phase-subsect}.  A possible first place to look is at SQL statements that require two phases of map-reduce, e.g., joins followed by aggregations.

\bibliographystyle{abbrv}
\bibliography{map-reduce}